\documentclass{article}
\usepackage{amsmath, amssymb, amsfonts, amsthm}
\usepackage[T1]{fontenc}
\usepackage{fullpage}
\usepackage{thmtools}
\usepackage{thm-restate}
\usepackage{algorithmicx}
\usepackage{algorithm}
\usepackage{algpseudocode}
\usepackage{graphicx}
\usepackage[normalem]{ulem}
\usepackage{cancel}
\usepackage[square,numbers]{natbib}
\usepackage{textcomp}
\usepackage{enumitem}
\usepackage{booktabs}
\usepackage{tikz}
\usepackage{bm}
\usepackage{standalone}
\usepackage{accents}
\usepackage{lineno}
\usepackage{graphicx}

\theoremstyle{plain}
\newtheorem{theorem}{Theorem}

\theoremstyle{definition}
\newtheorem*{example}{Example}

\makeatletter
\DeclareRobustCommand{\cev}[1]{%
	{\mathpalette\do@cev{#1}}%
}
\newcommand{\do@cev}[2]{%
	\vbox{\offinterlineskip
		\sbox\z@{$\m@th#1 x$}%
		\ialign{##\cr
			\hidewidth\reflectbox{$\m@th#1\vec{}\mkern4mu$}\hidewidth\cr
			\noalign{\kern-\ht\z@}
			$\m@th#1#2$\cr
		}%
	}%
}
\makeatother

\newcommand{\lvec}[1]{\overset{{}_{\leftarrow}}{#1}}

\newcommand\er{\mathsf{e}}
\newcommand\mr{\mathsf{m}}
\newcommand\el{\cev{\mathsf{e}}}
\newcommand\BWT{\mathsf{BWT}}
\newcommand\SA{\mathsf{SA}}
\newcommand\runs{\mathsf{r}}
\newcommand\dd{\,..\,}

\title{Relating Left and Right Extensions of Maximal Repeats}

\author{
	\textsc{Shunsuke Inenaga}\\
	\normalsize Kyushu University, Fukuoka, Japan\\
	\normalsize \texttt{inenaga.shunsuke.380@m.kyushu-u.ac.jp}
	\and
	\textsc{Dmitry Kosolobov}\\
	\normalsize Ural Federal University, Ekaterinburg, Russia\\
	\normalsize \texttt{dkosolobov@mail.ru}
}

\date{}

\begin{document}

\maketitle

\begin{abstract}
The compact directed acyclic word graph (CDAWG) of a string $T$ is an index occupying $O(\er)$ space, where $\er$ is the number of right extensions of maximal repeats in $T$. For highly repetitive datasets, the measure $\er$ typically is small compared to the length $n$ of $T$ and, thus, the CDAWG serves as a compressed index. Unlike other compressibility measures (as LZ77, string attractors, BWT runs, etc.), $\er$ is very unstable with respect to reversals: the CDAWG of the reversed string $\lvec{T} = T[n] \cdots T[2] T[1]$ has size $O(\el)$, where $\el$ is the number of left extensions of maximal repeats in $T$, and there are strings $T$ with $\frac{\el}{\er} \in \Omega(\sqrt{n})$. In this note, we prove that this lower bound is tight: $\frac{\el}{\er} \in O(\sqrt{n})$. Furthermore, given the alphabet size $\sigma$, we establish the alphabet-dependent bound $\frac{\el}{\er} \le \min\{\frac{2n}{\sigma}, \sigma\}$ and we show that it is asymptotically tight.

\end{abstract}

\algtext*{EndIf}
\algtext*{EndWhile}
\algtext*{EndFor}

\section{Introduction}

The indexing problem is to construct on a given string $T$ a data structure that, for any string $P$, can find all occurrences of $P$ in $T$. The development of lightweight string indexes has always been one of the central topics in string processing. There are numerous solutions for this problem involving so-called compressed and compact indexes: for a survey, see~\cite{Navarro, Navarro3,Navarro4} and references therein. One of such popular indexes is the compact directed acyclic word graph (CDAWG)~\cite{BelazzouguiCunial,BCGPR,BlumerEtAl2,BlumerEtAl,InenagaEtAl,TakagiEtAl}.

The CDAWG for the input string $T$ can be stored in $O(\er)$ space~\cite{BelazzouguiCunial,Inenaga}, where $\er$ is the number of right extensions of maximal repeats in $T$ (a precise definition is given below). For highly repetitive data, the measure $\er$ is typically small compared to the length $n$ of $T$ \cite{Navarro3,RadoszewskiRytter,Rytter} and, thus, the CDAWG serves as a compressed index. However, $\er$ is not particularly well behaved and, in a sense, it is the weakest compressibility measure compared to, for instance, LZ77~\cite{LZ77}, string attractors~\cite{KempaPrezza}, BWT runs~\cite{BurrowsWheeler}, and several others (see~\cite{Navarro3}). As an example of the bad behaviour, $\er$ is very unstable with respect to reversals (unlike other ``better'' measures): the CDAWG of the reversed string $\lvec{T} = T[n]\cdots T[2]T[1]$ has size $\el$, where $\el$ is the number of left extensions of maximal repeats in $T$, and there is a series of strings $T$ for which $\frac{\el}{\er} \in \Omega(\sqrt{n})$. 

In this note we show that the lower bound $\frac{\el}{\er} \in \Omega(\sqrt{n})$ is tight: $\frac{\el}{\er} \in O(\sqrt{n})$. Furthermore, given the alphabet size $\sigma$, we derive the alphabet-dependent bound $\frac{\el}{\er} \le \min\{\frac{2n}{\sigma}, \sigma\}$ and we describe a series of strings that asymptotically attain this bound, thus proving that it is tight. Observe that the inequality $\frac{\el}{\er} \in O(\sqrt{n})$ follows from the alphabet-dependent upper bound if one sets $\sigma \in \Theta(\sqrt{n})$. Our proofs are surprisingly simple.

Our tight upper and lower bounds for $\frac{\el}{\er}$ imply that improving space bounds from $O(\er+\el)$ to $O(\er)$ can shave a $\Theta(\sqrt{n})$ factor from the space requirements of some CDAWG-based data structures.
Such space improvements are known for storing the CDAWG ($O(\er+\el)$ space in~\cite{TakagiEtAl} and $O(\er)$ space in~\cite{BelazzouguiCunial,Inenaga}) and finding \emph{minimal absent words} ($O(\er+\el)$ space in~\cite{BelazzouguiC15} and $O(\min\{\er,\el\})$ space in~\cite{InenagaMAFF24}).

\subparagraph{Preliminaries.}
For a string $T = c_1 c_2 \cdots c_{n}$, denote by $|T|$ its length $n$. The \emph{reverse} $c_{n}\cdots c_2 c_1$ of $T$ is denoted $\lvec{T}$. We write $T[i]$ for the letter $c_i$ and $T[i \dd j]$ for the \emph{substring} $c_i c_{i + 1} \cdots c_j$, assuming $T[i \dd j]$ is empty if $i > j$. The \emph{empty string} is denoted by $\varepsilon$. We say that a string $S$ \emph{occurs} in $T$ at position $i$ if $T[i\dd i{+}|S|{-}1] = S$. A substring $T[i\dd n]$ is called a \emph{suffix} of $T$ and a substring $T[1\dd i]$ is called a \emph{prefix}. We denote integer segments by $[i\dd j] = \{i, i{+}1, \ldots, j\}$.

A string $S$ is called a \emph{maximal repeat} of a string $T$ if $S$ has at least two occurrences in $T$ and, for each letter $a$, the strings $aS$ and $Sa$ have strictly less occurrences in $T$. Note that the empty string is always a maximal repeat. For a substring $S$ of $T$, a letter $a$ is called a \emph{right} (respectively, \emph{left}) \emph{extension} of $S$ if $Sa$ (respectively, $aS$) occurs in $T$. It is easy to see that a string $S$ that occurs in $T$ at least twice is a maximal repeat of $T$ iff either $S$ is a prefix or it has at least two left extensions, and either $S$ is a suffix or it has at least two right extensions.
Let $M(T)$ denote the sets of maximal repeats in a string $T$.

Throughout the paper, our main focus will be on the number of right and left extensions of maximal repeats in string $T$: respectively, $\er_T = \sum_{S\in M(T)} r_S$ and $\el_T = \sum_{S\in M(T)} \ell_S$, where $r_S$ and $\ell_S$ are respectively the numbers of right and left extensions of $S$ in $T$. Observe that $\el_T = \er_{\lvec{T}}$. We use the notation $\er$ and $\el$ when the index string $T$ is clear from the context.

\begin{example}
	Consider the string $T = {\scriptstyle\diamondsuit}a{\scriptstyle\heartsuit}ab{\scriptstyle\clubsuit}abc{\scriptstyle\spadesuit}abcd$. Its set of maximal repeats is $\{a, ab, abc, \varepsilon\}$. The non-empty maximal repeats have the following respective sets of right extensions (the empty repeat $\varepsilon$ has every letter as its left and right extension): $\{{\scriptstyle\heartsuit},b\}, \{{\scriptstyle\clubsuit},c\}, \{{\scriptstyle\spadesuit},d\}$; and they have the following respective sets of left extensions: $\{{\scriptstyle\diamondsuit},{\scriptstyle\heartsuit},{\scriptstyle\clubsuit},{\scriptstyle\spadesuit}\}, \{{\scriptstyle\heartsuit},{\scriptstyle\clubsuit},{\scriptstyle\spadesuit}\}, \{{\scriptstyle\clubsuit},{\scriptstyle\spadesuit}\}$. Thus, the number of right and left extensions of maximal repeats is $\er = 2+2+2+8 = 14$ and $\el = 4+3+2+8 = 17$.
\end{example}

\section{Upper and Lower Bounds}

\begin{theorem}
For any string $T$ of length $n$, we have $\frac{\el}{\er} \le \min\{\frac{2n}{\sigma}, \sigma\}$, where $\sigma$ is the alphabet size and $\er$ and $\el$ denote the number of right and left extensions, respectively, of all maximal repeats in $T$.\label{thm:upper}
\end{theorem}
\begin{proof}
The observation that $\mr \le \er \le \mr\sigma$ and $\mr \le \el \le \mr\sigma$, where $\mr$ is the number of maximal repeats in $T$, implies the bound $\frac{\el}{\er} \le \sigma$. Since the empty maximal repeat has $\sigma$ right extensions, we have $\er \ge \sigma$. Thus, we obtain $\er \frac{n}{\sigma} \ge n$, which implies the inequality $\el < 2\er\frac{n}{\sigma}$ because $\el < 2n$ (see~\cite{BlumerEtAl,CrochemoreVerin2,CrochemoreVerin}). Thus, we obtain the bound $\frac{\el}{\er} \le \frac{2n}{\sigma}$, which leads to the result when combined with $\frac{\el}{\er} \le \sigma$.
\end{proof}

When $\sigma \in \Theta(\sqrt{n})$, Theorem~\ref{thm:upper} gives the upper bound $\frac{\el}{\er} \in O(\sqrt{n})$, which coincides with the known lower bound $\frac{\el}{\er} \in \Omega(\sqrt{n})$ shown in~\cite{Inenaga}. Namely, in~\cite{Inenaga} the following example was described (with reference to~\cite{KarkkainenPers}) that attains this bound over the alphabet $\{a_1,\ldots,a_k, \$_1, \ldots, \$_k\}$:
\begin{equation}
T_k = \$_1 a_1 \$_2 a_1 a_2 \$_3 a_1 a_2 a_3 \$_4 a_1 a_2 a_3 a_4 \$_5 \,\,\cdot\,\,\cdot\,\,\cdot\,\, \$_k a_1 a_2 \cdots a_k.\label{eq:example}
\end{equation}
The set of maximal repeats of $T_k$ is $\{\varepsilon, a_1, a_1a_2, \ldots, a_1a_2\cdots a_{k-1}\}$. Each of the non-empty repeats $a_1a_2\cdots a_i$ has exactly two right extensions $\$_{i+1}$ and $a_{i+1}$, and exactly $k - i + 1$ left extensions $\$_i, \$_{i+1}, \ldots, \$_k$. Therefore, we obtain $\er \in O(k)$ and $\el \ge \sum_{i=1}^{k-1} (k - i + 1) \in \Theta(k^2)$. It is easy to see that $n = |T_k| \in \Theta(k^2)$ and, thus, $\frac{\el}{\er} \in \Theta(k) = \Theta(\sqrt{n})$.

The alphabet size of $T_k$ is $\sigma = 2k \in \Theta(\sqrt{n})$.
Now, to complete the picture, we show that the alphabet-dependent upper bound of Theorem~\ref{thm:upper} is asymptotically tight for any range of the alphabet size.

\begin{theorem}
	For any integers $n > 1$ and $\sigma$ such that $1 < \sigma \le n$, one can construct a string $T$ of length $\Theta(n)$ over an alphabet of size $\Theta(\sigma)$ for which the following inequality holds: $\frac{\el}{\er} \in \Theta(\min\{\frac{n}{\sigma}, \sigma\})$, where $\er$ and $\el$ denote the number of right and left extensions, respectively, of all maximal repeats in $T$.\label{thm:lower} 
\end{theorem}
\begin{proof}

Consider the string
$$T = b_1 a^k \$ b_2 a^k \$ b_3 a^k \$  \cdots b_\sigma a^k \$$$
over the alphabet $\{\$, a, b_1, \ldots, b_\sigma\}$ of size $\Theta(\sigma$). Since $|T| \in \Theta(k\sigma)$, we obtain $|T| \in \Theta(n)$ by setting $k = \lceil n / \sigma\rceil$. The set of maximal repeats in $T$ is $\{\varepsilon, a, aa, \ldots, a^{k-1}, a^k\$\}$. Each $a^i$, for $1 \leq i < k$, has $\sigma+1$ left extensions $a, b_1, b_2, \ldots, b_\sigma$ and only two right extensions $\$$ and $a$. The repeat $a^k\$$ has $\sigma$ left extensions $b_1, \ldots, b_\sigma$ and $\sigma - 1$ right extensions $b_2, \ldots, b_{\sigma}$. The empty maximal repeat $\varepsilon$ has $\sigma+2$ left and right extensions. Overall, $T$ has $\el \in \Theta(k \sigma) = \Theta(n)$ left extensions and $\er \in \Theta(k + \sigma)$ right extensions.

\begin{itemize}
\item If $\sigma \geq \sqrt{n}$, then since $k = \lceil n / \sigma\rceil$, we obtain $k \le \sigma$.
Hence $\er \in \Theta(k+\sigma) = \Theta(\sigma)$.
Thus $\frac{\el}{\er} \in \Theta(\frac{n}{\sigma}) = \Theta(\min\{\frac{n}{\sigma}, \sigma\})$ holds for the string $T$.

\item If $\sigma < \sqrt{n}$, then since $k = \lceil n / \sigma\rceil$, we obtain $k \ge \sigma$.
  Hence, $\er \in \Theta(k + \sigma) = \Theta(k) = \Theta(\frac{n}{\sigma})$. Thus, $\frac{\el}{\er} \in \Theta(\sigma) = \Theta(\min\{\frac{n}{\sigma}, \sigma\})$ holds for the string $T$.
\end{itemize}
This completes the proof.
\end{proof}

\paragraph{Acknowledgement} 
The authors wish to thank Wiktor Zuba for initial discussions of the problem at StringMasters 2024 in Fukuoka, Japan.

\bibliographystyle{abbrv}
\bibliography{refs}

\end{document}